\documentclass{article}

\usepackage{geometry}
\usepackage{hyperref}
\usepackage{lmodern}
\usepackage{animate}
\usepackage{tikz}
\usepackage{tikz-cd}
\usetikzlibrary{shapes.misc, positioning}
\usetikzlibrary{arrows,shapes}
\usetikzlibrary{babel}
\usepackage{mathrsfs,amsmath}
\usepackage{amsthm}
\usepackage{physics}
\usepackage{amsfonts}
\usepackage{amssymb}
\usepackage{graphicx}
\usepackage{wrapfig}
\usepackage{cancel}
\usepackage{braket}
\usepackage[all]{xy}
\usepackage{xcolor}
\usepackage{adjustbox}
\usepackage{subcaption}
\usepackage{capt-of}
\usepackage{tabu}
\makeatletter
\def\fdsy@scale{1.8}
\newcommand\fdsy@mweight@normal{Book}
\newcommand\fdsy@mweight@small{Book}
\newcommand\fdsy@bweight@normal{Medium}
\newcommand\fdsy@bweight@small{Medium}
\DeclareFontFamily{U}{FdSymbolC}{}
\DeclareFontShape{U}{FdSymbolC}{m}{n}{
    <-7.1> s * [\fdsy@scale] FdSymbolC-\fdsy@mweight@small
    <7.1-> s * [\fdsy@scale] FdSymbolC-\fdsy@mweight@normal
}{}
\DeclareFontShape{U}{FdSymbolC}{b}{n}{
    <-7.1> s * [\fdsy@scale] FdSymbolC-\fdsy@bweight@small
    <7.1-> s * [\fdsy@scale] FdSymbolC-\fdsy@bweight@normal
}{}
\makeatother
\DeclareFontFamily{U}{FdSymbolA}{}
\DeclareFontShape{U}{FdSymbolA}{m}{n}{<->FdSymbolA-Book}{}
\DeclareSymbolFont{extrasymbols}{U}{FdSymbolA}{m}{n}
\DeclareMathSymbol{\vardiamondsuit}{\mathord}{extrasymbols}{182}
\DeclareMathSymbol{\varheartsuit}{\mathord}{extrasymbols}{184}

\definecolor{color1}{RGB}{12,18,155}
\definecolor{color2}{RGB}{83,89,155}
\definecolor{color3}{RGB}{75,80,85}
\newtheorem{theorem}{Theorem}[section]

\newtheorem{proposition}[theorem]{Proposition}

\begin{document}

\title{\textbf{Relative entropy of an interval for a massless boson at finite temperature}}

\author{Alan Garbarz$^\diamondsuit{}^\perp$, Gabriel Palau$^\diamondsuit$}
\date{}
\maketitle

\vspace{.25cm}

\begin{minipage}{.9\textwidth}\small \it 
	\begin{center}
    $^\diamondsuit$  Universidad de Buenos Aires, Facultad de Ciencias Exactas y Naturales, Departamento de Física. Ciudad Universitaria, pabell\'on 1, 1428, Buenos Aires, Argentina.
     \end{center}
\end{minipage}

\vspace{.25cm}

\begin{minipage}{.9\textwidth}\small \it 
	\begin{center}
    $^\perp$ CONICET - Universidad de Buenos Aires, Instituto de Física de Buenos Aires (IFIBA). Ciudad Universitaria, pabell\'on 1, 1428, Buenos Aires, Argentina.
     \end{center}
\end{minipage}
\vspace{.5cm}

\maketitle

\begin{abstract}
    We  compute  Araki's relative entropy associated to a bounded interval $I=(a,b)$ between a thermal state and a coherent excitation of itself in the bosonic U(1)-current model, namely the (derivative 
    of the) chiral boson. For this purpose we briefly review some recent results on the entropy of standard subspaces and on the relative entropy of non-pure states such as thermal states. In particular, recently Bostelmann, Cadamuro and Del Vecchio have obtained the relative entropy at finite temperature for the unbounded interval $(-\infty,t)$, using  previous results of Borchers and Yngvason, mainly a unitary dilation that provides the modular evolution in the negative half-line. Here we find a unitary rotation in order to make use of the full PSL$(2,\mathbb{R})$ symmetries and obtain the modular group, modular Hamiltonian and the relative entropy  $S$ of a bounded interval at finite temperature. Such relative entropy entails both a Bekenstein-like bound and a QNEC-like bound, but violates  $S''\geq 0$. Finally, we extend the results to the free massless boson in $1+1$ dimensions with analogous bounds.  
    \begin{flushleft}
\hrulefill\\
\footnotesize
{E-mails:  alan@df.uba.ar, gpalau@df.uba.ar}
\end{flushleft}
\end{abstract}

\pagebreak

\tableofcontents

\pagebreak

\section{Introduction}

The relative entropy is a measure of the distinguishability of  two states. In quantum mechanics, for two density matrices $\rho$ and $\sigma$ it is defined as
\begin{equation}
    S(\rho||\sigma)=\text{Tr}(\rho \log \rho-\rho \log \sigma)
    \label{QMrelativeentropy}
\end{equation}
If the system is bi-partite, one can reduce each density matrix to one of the partitions and compute the relative entropy between the reduced density matrices. The counterpart in QFT would be to reduce the states to a spacetime region. Intuitively, the smallest this region is, the lesser operators one has at hand to characterise the states and therefore the relative entropy decreases. However, generic bounded regions of spacetime such as causal diamonds are assigned a von Neumann algebra which is a Type III factor, and these algebras have no trace-class operators and there are no density matrices \cite{haag1996local}. Nevertheless, the definition of the relative entropy \eqref{QMrelativeentropy} can be suitably extended \cite{Araki:1973hh}: given two (normal and faithful) states $\omega$ and $\omega'$ of the von Neumann algebra $A$ associated to some region, with both states represented on a Hilbert space by $\Omega$ and $\Omega'$, the relative entropy is
\begin{equation}
        S_{A}(\omega||\omega')=-(\Omega,\log\Delta_{\Omega',\Omega}\,\Omega),
    \label{Arakientropy}
    \end{equation}
where $\Delta_{\Omega',\Omega}$ is the relative modular operator (which is defined in the next section). This is in sharp contrast with the entanglement entropy which necessarily diverges due to the generic UV behavior of correlations through the boundary of the corresponding region \cite{Witten:2018zxz}. 

Araki's relative entropy \eqref{Arakientropy} can be hard to compute in general cases. Recently it has been computed in a number of specific situations for coherent states acting on the vacuum of the free scalar QFT: on a Rindler wedge \cite{Longo:2019mhx,Ciolli:2019mjo,Casini:2019qst}, on a causal diamond in spacetime dimension greater than 2 \cite{Longo:2020amm}, and later the two-dimensional case  \cite{Longo:2021rag}. In addition, in the bosonic  vacuum U(1)-current model (namely, a massless boson on a light ray) the relative entropy of coherent states was computed for a half-line in \cite{Longo:2018obd} and used to obtain the analogous expression for (unbounded) regions of the null plane 
for a free scalar in \cite{Morinelli:2021nsx}. For similar results in free fermionic CFTs, see \cite{Longo:2017mbg}. For interacting theories, as far as we know the only results available correspond to coherent states in chiral CFTs \cite{Hollands:2019czd,Panebianco:2019plp}, while for  free QFT in curved spacetimes see \cite{Ciolli:2021otw}. 

All of the above computations compare the vacuum state to a coherent excitation of itself (with the exception of \cite{Ciolli:2021otw}). Recently in \cite{Bostelmann:2020srs}, among other things, the relative entropy for a half-line in the \textit{thermal} U(1)-current model was computed. That is, a KMS state of inverse temperature $\beta$ was compared to a coherent excitation of itself, when restricted to a half-line. The fact that the relative entropy was computed on the half-line is an important point, since by general arguments its modular operator had already been obtained in  \cite{Borchers:1998ye}. Even in the vacuum case, one can see that to obtain the modular operator for a bounded interval it is necessary to make use of the full PSL$(2,\mathbb{R})$  group of conformal transformations of the chiral boson theory (we will discuss this later in detail). 

The main goal 
of this article is to compute the relative entropy of coherent states on a bounded interval $I \subset \mathbb{R}$ for the thermal U(1)-current model and the free massless boson at finite temperature in $1+1$ dimensions restricted to a causal diamond. We will find two bounds for each theory, a Bekenstein-like bound and mainly a controlled violation of the Quantum Null Energy Condition (QNEC), which was already anticipated in \cite{Bostelmann:2020srs}.

\section{Preliminaries on the modular structure of the Weyl algebra}

Given a symplectic space $(\mathcal{K},\sigma)$ we get a CCR algebra with relations
    \begin{equation}
    W(f)W(g)=e^{-i\sigma(f,g)}W(f+g),\quad W(f)^*=W(-f)    
    \end{equation}
where $f,g \in \mathcal{K}$. We will refer to this algebra as CCR$(\mathcal{K},\sigma)$.

A positive symmetric bilinear form $\tau$ such that 
\begin{equation}\label{technical}
    \sigma(f,g)^2\leq \tau(f,f)\tau(g,g)
\end{equation}
defines a quasi-free state by \cite{Kay:1988mu} 
\begin{equation}
    \omega(W(f))=e^{-\frac{1}{2}\tau(f,f)}
\end{equation}
with 2-point function $w_2(f,g)=\tau(f,g)+i\sigma(f,g)$. However the state may not be pure. It is pure if and only if $w_2$ is a complex inner product. More precisely, from \eqref{technical} it can be shown that a contraction $D$ exists such that 
\begin{equation}
    \sigma(f,g)=\tau(f,D g).
    \label{D}
\end{equation}
Because of the non-degeneracy of $\sigma$, $ D$ is invertible and $D=C|D|$ is its polar decomposition . $C$ is a complex structure and the state $\omega$ is pure if and only if $w_2$ is a complex inner product, bi-linear with respect to the complex structure $C$. This is actually equivalent to the statement that $\omega$ is pure if and only if $|D|=1$ (see \cite{Petz:1990gb} for further details). Let us differ how to construct a purification of a non-pure state to the end of this section.

\subsection{Modular theory and relative entropy}

We assume we have a complex Hilbert space $\mathcal{H}$ with inner product $\langle f,g\rangle = w_2(f,g)$. In other words, we are assuming for now that the state is pure.  Let us call its (bosonic) Fock space $\Gamma(\mathcal{H})$. We have a representation of the CCR$(\mathcal{K},\sigma)$ algebra on the Fock space $\Gamma(\mathcal{H})$. Indeed, $W(f)$ acts on $\Gamma(\mathcal{H})$ as $V(f)$:
    \begin{equation}
        V(f)e^0=e^{-\frac{1}{2}\langle f,f\rangle}e^f,\qquad f \in \mathcal{H} 
    \end{equation}
with
\begin{equation}
    e^f:=1\oplus f\oplus\frac{1}{\sqrt{2!}} f \otimes f \oplus ...
\end{equation}
Calling $\Omega:=e^0$ the vacuum vector
\begin{equation}
    (\Omega,V(f)\Omega)=e^{-\frac{1}{2}\langle f,f\rangle}=\omega(W(f))
\end{equation}
We can define the local algebras associated to a given real-linear subspace $H\subset \mathcal{H}$:
\begin{equation}
    R(H):=\left\{V(f);\quad f\in H \right\}''
\end{equation}
It turns out that $\Omega$ is cyclic (respectively separating) for $R(H)$ if and only if $H$ is cyclic (respectively separating). $H$ is cyclic if $\overline{H+iH}=\mathcal{H}$ while separating if $H\cap iH=0$. If $H$ is also  closed it is called a \textit{standard} subspace .      
    
The Tomita operator $S$ associated to $R(H)$ (and $\Omega$) is defined by (the closure of)
    \begin{equation}
        S V(f) \Omega=V(f)^*\Omega,\quad  V(f)\in R(H) 
    \end{equation}
The \textit{relative} Tomita operator $S_{\Omega',\Omega}$ associated to $R(H)$ is defined by (the closure of)
    \begin{equation}
        S_{\Omega',\Omega} V(f) \Omega=V(f)^*\Omega'
    \end{equation}
with polar decomposition
\begin{equation}
    S_{\Omega',\Omega}=J_{\Omega',\Omega}\Delta_{\Omega',\Omega}^{1/2}
\end{equation}

For some algebra $R(H)$ and cyclic and separating states $\omega$ and $\omega'$, Araki's relative entropy is defined as
    \begin{equation}
        S_{A}(\omega||\omega')=-(\Omega,\log\Delta_{\Omega',\Omega}\,\Omega)
    \label{Arakirelativeentropy}
    \end{equation}
This is hard to compute for generic cases, however it  simplifies considerably for coherent states\footnote{The relative entropy between coherent states satisfies $S(\omega_f||\omega_g)=S(\omega_{f-g}||\omega)$, so there is no loss of generality in assuming that one state is not excited by a Weyl unitary.}, namely when $\Omega'=V(f)\Omega$, and we shall call the corresponding algebraic state $\omega_f$. In order to see this, we first need to introduce a modular theory for $H\subset \mathcal{H}$, with  $H$  \textit{standard}, following \cite{Longo08}. 
The analogous Tomita operator is defined by the closure of $S_H(f+ig)=f-ig$ and its polar decomposition is
\begin{equation}
    S_H=J_H\Delta_H^{1/2}
\end{equation}
Then, the entropy of a vector $f\in H$ w.r.t. $H$ is defined by
\begin{equation}
        S_H(f):=-\langle f,\log\Delta_H f\rangle  
\end{equation}
Actually, in \cite{Ciolli:2019mjo} it  was generalized for $f\in \mathcal{H}\cap \text{Dom}(K_H)$, 
\begin{equation}
    S_H(f):=-\langle f,P_H\log\Delta_H f\rangle=\sigma(f,P_H i K_H f),
    \label{longoentropy}
\end{equation}
where 
\begin{equation}
    K_H:=-\log \Delta_H=i\frac{d}{du}\Delta^{iu}_H\bigg|_{u=0},
\end{equation}
is the 1-particle modular Hamiltonian and where $P_H:H+H'\rightarrow H$ is the cutting projector and we are assumming that $H$ is \textit{factorial}, namely $H \cap H'=0$. Here $H'$ is the symplectic complement of $H$. In that work they showed that Araki's relative entropy \eqref{Arakirelativeentropy}, between a coherent state $V(f)\Omega$ and a \textit{pure
} state $\Omega$, is nothing but the entropy of the vector $f\in \mathcal{H} \cap \text{Dom}(K_H)$:
\begin{equation}
    S(\omega_f||\omega)=\sigma(f,P_H i K_H f)
    \label{relativeentropy}
\end{equation}
We will take advantage of this result to compute the relative entropy by working exclusively at the level of the 1-particle Hilbert space $\mathcal{H}$.

\subsection{Purification of $\omega$}

Since we are interested in computing the relative entropy in the case that $\omega$ is a thermal state, and since \eqref{relativeentropy} is valid for pure states, we need to work with a purification of $\omega$. This can be achieved by a procedure we recall in this subsection, following \cite{Bostelmann:2020srs} (see also \cite{Petz:1990gb}).

By means of \eqref{D} we can get\footnote{We assume that $\mathcal{K}$ is complete with respect to $\tau$ and that $\sigma$ is non-degenerate.} a complex structure on $\mathcal{K}^\oplus:=\mathcal{K}\oplus \mathcal{K}$ 
\begin{equation}
    i^\oplus=\left(
    \begin{matrix}
    -D& C\sqrt{1+D^2}\\
    C\sqrt{1+D^2} & D
    \end{matrix}\right)
    \label{iplus}
\end{equation}
Let us call $\mathcal{H}^\oplus$ the complexification of $\mathcal{K}^\oplus$. The complex inner product is given by
    \begin{equation}
        \langle \cdot ,\cdot\rangle^\oplus=\tau^\oplus(\cdot,\cdot ) + i \sigma^\oplus(\cdot,\cdot ) ,
    \end{equation} 
with $\tau^\oplus:=\tau \oplus \tau$ and $\sigma^\oplus(\cdot,\cdot)=\tau^\oplus(\cdot,-i^\oplus \cdot)$. This inner product  reduces to $w_2$ on $\mathcal{K}\simeq \mathcal{K}\oplus 0$:
\begin{equation}
    \langle f\oplus 0,g\oplus 0\rangle^\oplus=\tau(f,g)+i\sigma(f,g)=w_2(f,g)
    \end{equation}
We have 
    \begin{equation}
        \text{CCR}(\mathcal{K},\sigma)\subset \text{CCR}(\mathcal{K}^\oplus,\sigma^\oplus) 
    \end{equation} 
Similarly, for closed subspaces $H\subset \mathcal{K}$:
    \begin{equation}
        \text{CCR}(H,\sigma)\subset \text{CCR}(\mathcal{K},\sigma) 
    \end{equation} 
And most importantly on the CCR$(\mathcal{K},\sigma)$ algebra,
    \begin{equation}
        \omega^\oplus(W(f\oplus 0))=e^{-\frac{1}{2}\tau^\oplus(f\oplus 0,\,f\oplus 0)}=\omega(W(f))
    \end{equation}
which justifies why the pure state $\omega^\oplus$ of the  CCR$(\mathcal{K}^\oplus,\sigma^\oplus)$ algebra is a purification of $\omega$. 

The relative entropy for non-pure states associated to $R(H)$, with $H\simeq H \oplus 0$ the standard and factorial subspace of $\mathcal{H}^\oplus$ reads \cite{Bostelmann:2020srs}, 
\begin{equation}
    S_{R(H)}(\omega_f||\omega)=\sigma^\oplus(f ,P_H  i^\oplus  K_H f),\qquad f\in \mathcal{K}\oplus 0 \cap \text{Dom}(K_H).
    \label{relativeentropypurified}
\end{equation}
In \cite{Bostelmann:2020srs} a more general expression was obtained for other subspaces. We will not need this since for the $U(1)$ model the subspace associated to the bounded interval is standard and factorial as we will show. Note that in \eqref{relativeentropypurified} the modular Hamiltonian $K_H$ is associated to the modular operator that acts on the larger space $\mathcal{H}^\oplus$ while $P_H$ is the the real-linear cutting projector onto $H \oplus 0$.

\section{The chiral boson current}

We are interested in applying all the above to the case of a free chiral boson. More precisely, we consider the current usually denoted $\phi'(x)$. In the smeared version, the classical theory is defined by the symplectic space of compactly supported real functions  $\mathcal{K}=C_c^\infty(\mathbb{R})$ with symplectic structure
\begin{equation}
    \sigma(f,g)= \int_{\mathbb{R}} f(x)g'(x)dx
    \label{symplectic}
\end{equation}
As reviewed in the previous section, we have a CCR$(\mathcal{K},\sigma)$ algebra associated to $(\mathcal{K},\sigma)$. In order to proceed, we need to define a quasi-free state by means of a positive symmetric bilinear form $\tau$. We start with the vacuum state which we will denote $\omega$ and then move on to a thermal state $\omega_\beta$.

\subsection{The vacuum case}
The vacuum state is defined by
        \begin{equation}
            \tau(f,g)=-\frac{1}{\pi}PV\int_{\mathbb{R}^2}dx\,dy\frac{f(x)g(y)}{(x-y)^2}=(f,\mathfrak{H}g')_{L^2}
        \end{equation}
Here $PV$ denotes the principal value integral and $\mathfrak{H}$ the Hilbert transform. It is pure since $D=-\mathfrak{H}$ which is unitary or equivalently $D=-i \text{sgn}(p)$ in momentum space\footnote{We are taking the Fourier transform as $\hat{f}(p)= \int_{\mathbb{R}}dx\,e^{ixp}f(x)$. With this convention the Hilbert transform in momentum space is $i$ sgn$(p)$. More importantly, the generator of unitary translations \eqref{unitaryrep} is positive. \label{footnoteFourier}}, and therefore the complexification described in the previous section can be applied directly to $\mathcal{K}$ and we obtain the complex inner product
        \begin{align}
            \langle f,g \rangle&=-\frac{1}{\pi}\int_{\mathbb{R}^2}dx\,dy\frac{f(x)g(y)}{(x-y-i\epsilon)^2}=(f,\mathfrak{H}g')_{L^2}+i\sigma(f,g)\\
         &=\frac{1}{\pi}\int_{0}^\infty\hat{f}(p)^*\hat{g}(p)pdp
        \end{align}
which enables to establish an isomorphism with  $\mathcal{H}\simeq L^2(\mathbb{R}_+,pdp)$. This is the vacuum 1-particle Hilbert space. We will mostly work in coordinate space. 

An important role in this work is played by the PSL$(2,\mathbb{R})$ symmetry of the model. The unitary representation on $\mathcal{H}$ is given by,
        \begin{equation}
            U(g)f(x)=f(g^{-1}\cdot x)
        \end{equation}
with an element of the group $g$ acting on the coordinate $x$ by linear fractional transformations: 
\begin{equation}
            g\cdot x=\frac{ax+b}{cx+d},\qquad g\in \text{PSL}(2,\mathbb{R}).
\end{equation}
We have three one-parametric subgroups related to the KAN decomposition of PSL$(2,\mathbb{R})$
        \begin{equation}
            r(\theta)=\left(
    \begin{matrix}
    \cos\frac{\theta}{2} & \sin\frac{\theta}{2}\\
    -\sin\frac{\theta}{2} & \cos\frac{\theta}{2}
    \end{matrix}\right),\quad \delta(s)=\left(
    \begin{matrix}
    e^{\frac{s}{2}} & 0\\
    0 & e^{-\frac{s}{2}}
    \end{matrix}\right),\quad \tau(t)=\left(
    \begin{matrix}
    1 & t\\
    0 & 1
    \end{matrix}\right),  
        \end{equation} 
which we will refer to as rotatation, dilation and translation subroups, respectively. For example the dilation-translation subgroup acts as:
        \begin{equation}
            U(\tau(t))f(x)=f(x-t),\qquad U(\delta(s))f(x)=f(e^{-s}x)
            \label{unitaryrep}
        \end{equation}
Note that the generator  of translations $P$ defined by $U(\tau(t))=e^{itP}$ is positive. This is most easily seen in momentum-space with the convention of footnote \ref{footnoteFourier}. 

Now we can discuss the modular theory of standard subspaces of $\mathcal{H}$. Consider the subspaces $H(I)=\overline{C_c^\infty(a,b)}\in \mathcal{H}$, which are standard and factorial \cite{Longo08}. Since the above representation has positive $P$, the modular evolution on $H(\mathbb{R}_+)$ is (Theorem 3.3.1 of \cite{Longo08})
        \begin{equation}
            \Delta_{(0,\infty)}^{iu}=U\left(\delta\left(-2\pi u \right)\right),\quad u\in \mathbb{R} 
        \end{equation}
This can be seen by checking that  $F(u):=\langle g,\Delta^{iu} f\rangle$ admits an analytic continuation to the strip $-1<$ Im$(u)<0$  and that the KMS  property at temperature $-1$ is satisfied: 
\begin{equation}
    \langle\Delta_{(0,\infty)}^{iu}f,g\rangle=\langle g,\Delta_{(0,\infty)}^{i(u-i)}f\rangle
\end{equation}
By covariance $U(g)H(I)=H(g\cdot I)$ we have
        \begin{equation}
            \Delta_I^{iu}=U(\Bar{g})^{-1}\Delta_{(0,\infty)}^{iu}U(\Bar{g}),\quad u\in \mathbb{R},\quad \Bar{g}\cdot I= \mathbb{R}_+
            \label{Imodular}
        \end{equation} 
Note that $\Bar{g}$ is defined modulo multiplication by a dilation on the left, but this does not affect $\Delta_I^ {iu}$.

\subsubsection*{The interval $(-\infty,t)$}
For instance, by considering $\bar{g}$ a rotation in $\pi$ followed by a translation of $t$,
        \begin{equation}
            \Delta^{iu}_{(-\infty,t)}f(x)=f\left(e^{-2\pi u}(x-t)+t\right).
            \label{Bostelmannvacuumflow}
        \end{equation} 
We can compute then the modular Hamiltonian:
        \begin{equation}
            K_{(-\infty,t)}=i\frac{d}{du}\Delta_{(-\infty,t)}^{iu}\bigg|_{u=0}=2\pi i(t-x) \frac{d}{dx}
        \end{equation}
And from here the relative entropy of a coherent state:
    \begin{equation}
        S_{R((-\infty,t))}(\omega||\omega_f)=S_{(-\infty,t)}(f)=2\pi\int_{-\infty}^t (t-x)f'(x)^2 dx 
    \end{equation}
Here $f$ need not be localized in the interval $(-\infty,t)$ \cite{Bostelmann:2020srs}.

\subsubsection*{The interval $(a,b)$}

We can repeat what we have just done for an interval $I=(a,b)$. The first step is to find $\Bar{g}$ such that $\Bar{g}\cdot I=\mathbb{R}_-$. The idea is to first project $I$ to the circle, forming an arc $(\theta_a,\theta_b)$. Then, we employ two symmetries (see Figure \ref{figure:gbar}). The first one consists of a rotation that maps $I$ to $(-\infty,b')$, which can be achieved by noticing that in the circle this is just a rotation of magnitude $-\pi-\theta_a$. Then, $(-\infty,b')$ is mapped by a translation of magnitude $b'$ to $\mathbb{R}_-$.

\begin{figure}[h]
\centering
\includegraphics[width=11cm]{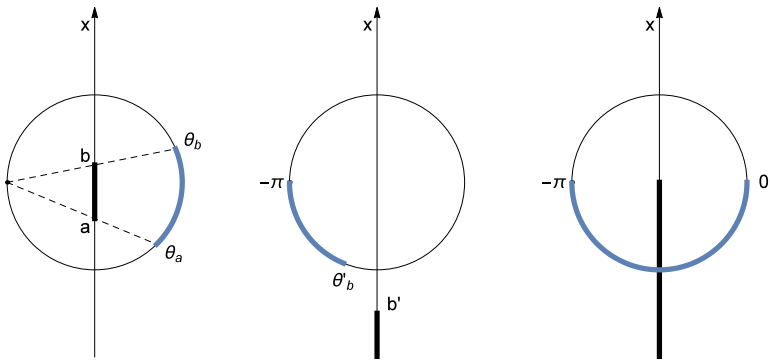}
\caption{In the left, a generic interval $(a,b)$ projected to the circle. It is first mapped to $(-\infty,b')$ by a rotation of magnitude $-\pi-\theta_a$, as shown in the middle diagram. Then it is mapped by a translation by $b'$ to $\mathbb{R}_-$, on the right.}
\label{figure:gbar}
\end{figure}

Having obtained such $\Bar{g}$, we can compute the modular evolution using \eqref{Imodular},
\begin{equation}
    \Delta^{iu}_{(a,b)}f(x)=f\left(\frac{a e^{-2\pi u}(b-x)+b (x-a)}{e^{-2\pi u} (b-x)+x-a} \right).
\end{equation}
The modular Hamiltonian is
     \begin{equation}
         K_{(a,b)}=\frac{2\pi i (x-a)(b-x)}{b-a}\frac{d}{dx}.
         \label{IModHam}
     \end{equation}
The last ingredient to compute the relative entropy is the cutting projector $P_{I}$ associated to the interval $I=(a,b)$. Following the same lines in the proof of Proposition 4.2 in \cite{Bostelmann:2020srs} one can see that $P_{I}K_{I}f(x)=\chi_{I}(x)K_{I}f(x)$ where $\chi_{I}$ is the characteristic function of the interval. 
\begin{proposition}
    Let $f\in C^{\infty}_{c}(\mathbb{R})$, then it holds that $P_{I}K_{I}f(x)=\chi_{I}(x)K_{I}f(x)$ 
\end{proposition}
\begin{proof}
     Given $f\in C^{\infty}_{c}(\mathbb{R})$ define the functions $g(x)=\chi_{I}(x)K_{I}f(x)$ and $g_{c}(x)=\chi_{I^{c}}(x)K_{I}f(x)$, which are piecewise-differentiable functions so their Fourier transforms decay at least like $p^{-2}$ for large $p$ and 
\begin{equation*}
    \|g\|_{\tau}^{2}=\frac{1}{\pi}\Re{\int_{0}^{\infty}|\hat{g}(p)|^2 pdp}<\infty,
\end{equation*}     
then $g\in \mathcal{H}$ and also $g_{c}\in \mathcal{H}$. Moreover $\sigma(g_{c},\varphi)=0$ for all $\varphi \in C_{c}^{\infty}(I)$ because supp$(g_{c})\subseteq I_{c}$, then by continuity of $\sigma$ (with respect to the topology induced by $\tau$),  $g_{c}\in H(I)'$. Similarly one can see that $g\in H(I)$ and therefore, as $H(I)$ is factorial, $P_{I}K_{I}f(x)=P_{I}(g(x)+g_{c}(x))=g(x)=\chi_{I}(x)K_{I}f(x)$, which completes the proof. 
\end{proof}

 Finally the relative entropy is
\begin{equation}
    S_{R(I)}(\omega||\omega_f)=S_{I}(f)=2\pi\int_{a}^b \frac{ (x-a)(b-x)}{b-a}f'(x)^2 dx
    \label{vacuumrelativeentropy}
     \end{equation} 
It is translation invariant, in the sense that  $S_{I}(f)=S_{\tau\cdot I}(f\circ \tau^{-1})$. Is is also immediate to see that it is increasing with $L=(b-a)$: $\frac{d}{dL}S_{I(L)}(f)>0$. Interestingly, the relative entropy \eqref{vacuumrelativeentropy} satisfies a  Bekenstein-like bound\footnote{Here $E(f):=\int_a^b f'(x)^2 dx$ is the 1-particle energy associated to the interval of $f \in \mathcal{H}$. } 
    \begin{equation}
       S_{I}(f)\leq \pi \frac{L}{2} \int_a^b f'(x)^2\,dx =:\pi \frac{L}{2} E(f)
        \label{Bekensteinvacuum}
    \end{equation}
    and a QNEC-like bound\footnote{This bound (for the interval centered at $0$) is the one in Proposition 3.7 of \cite{Bostelmann:2020srs}, where the appropriate $T_{f}(s,L)$ is $2\pi \int_{-\min\left\{s,\frac{L}{2}\right\}}^{\min\left\{s,\frac{L}{2}\right\}}\dfrac{(x+\frac{L}{2})(\frac{L}{2}-x)}{L}f'(x)^2\,dx$. It is straightforward to see that such $T_{f}$ satisfies the smoothness hypothesis $C^1$ of the Proposition. }: 
    \begin{align}
      S_I'':=   \frac{d^2}{dL^2}S_{I(L)}(f)&=\frac{\pi}{2}\left(f'(b)^2+f'(a)^2\right)-\frac{4\pi}{L^3}\int_a^b \left(x-\frac{a+b}{2}\right)^2f'(x)^2\,dx\nonumber\\
        &\geq -\frac{4\pi}{L^3}\int_a^b \left(x-\frac{a+b}{2}\right)^2f'(x)^2\,dx.     \label{QNECvacuum}
    \end{align}
The QNEC, when stated in terms of the relative entropy, reads $S''(\lambda)>0$, with the understanding that $\lambda$ continuously labels  nested spacetime regions. However in \eqref{QNECvacuum} we see a  violation of the QNEC, which was anticipated in \cite{Bostelmann:2020srs}. What \eqref{QNECvacuum} says is that in order to have a large violation of the QNEC, a considerable amount of energy must be concentrated near the boundaries of the interval (note also that this negative bound can be saturated). Similarly, the Bekenstein-like bound \eqref{Bekensteinvacuum} implies that in order to make a coherent state largely distinguishable from the vacuum,  a considerable amount of energy needs to be placed in the interval.    

\subsection{The thermal case}

We now turn our attention to thermal states. The underlying symplectic space is again $(C_c^\infty(\mathbb{R}),\sigma)$ with \eqref{symplectic}. The thermal state is defined by 
        \begin{equation}
            \tau_\beta(f,g)=-\pi PV\int_{\mathbb{R}^2}dx\,dy\frac{f(x)g(y)}{\beta^2\sinh^2\left(\frac{\pi}{\beta}(x-y)\right)}
        \end{equation}
This gives the 2-point function \cite{Borchers:1998ye}, 
        \begin{equation}
            w_2^{(\beta)}( f,g)=-\pi\int_{\mathbb{R}^2}\frac{f(x)g(y)}{\beta^2\sinh^2\left(\frac{\pi}{\beta}(x-y)-i\epsilon\right)}dx\,dy.
        \end{equation}
Such 2-point function satisfies being translation invariant and the KMS condition with respect to translations.  The real Hilbert space $\mathcal{K}=L^2(\mathbb{R}_+,\frac{pdp}{1-e^{-\beta p}})$ is obtained after completion of $C_c^\infty(\mathbb{R})$ with $\tau_\beta$ \cite{Bostelmann:2020srs}. Note that this thermal state is the geometric KMS state of \cite{Longo:2016pci,Camassa:2011wk}\footnote{We thank Yoh Tanimoto for pointing this out.}.

Since the state is not pure, we first proceed to ``purify''. In momentum space $D=-i(1-e^ {-\beta p})$, which we use to construct $i^\oplus$ (as given by \eqref{iplus}) and then $\mathcal{H}^\oplus$, the complexification of $\mathcal{K}\oplus\mathcal{K}$. 
            
There are operators acting as the dilation-translation group\footnote{We are not claiming these operators form a representation on $\mathcal{K}$. For instance, there are values of the parameters where the logarithms are not defined and the function is instead defined to be zero. We will not need to take this into account.} on the half-lines. For example, on $H(\mathbb{R}_-)$ \cite{Borchers:1998ye}:
        \begin{align}
            U_\beta(\delta(s))f(x)&=f\left( -\frac{\beta}{2\pi}\log\left(1+e^{-s}(e^{-\frac{2\pi x}{\beta}}-1)\right) \right),\nonumber\\
            U_\beta(\tau(t))f(x)&=f\left(x- \frac{\beta}{2\pi}\log\left(1+\frac{2\pi t}{\beta}e^{\frac{2\pi x}{\beta}}\right) \right),
            \label{BYoperators}
        \end{align}
        which satisfy
        \[ U_\beta(\delta(s))U_\beta(\tau(t))U_\beta(\delta(-s))=U_\beta(\tau(e^st)).\]
These operators leave $w_2^{(\beta)}$ invariant, which implies that
         \begin{equation}
             U_\beta^\oplus (f \oplus 0 + i^\oplus g \oplus 0):=U_\beta f \oplus 0 + i^ \oplus U_\beta g \oplus 0, \qquad f,g\in\mathcal{K}, \
             \label{Uplus}
         \end{equation}
are unitaries of $\mathcal{H}^\oplus$ (we show this later on). Because of this, the modular operator associated to $H(\mathbb{R}_-)$ is given by\footnote{We have a sign difference with respect to \cite{Bostelmann:2020srs} in the parameter inside the dilation $\delta$. This translates into a sign difference in the modular Hamiltonian, but then the relative entropy coincides with equation (5.22) of that reference. } \cite{Bostelmann:2020srs},       
\begin{equation}
    \Delta^{iu}_{H(\mathbb{R}_-)}=U^\oplus_\beta(\delta(2\pi u)).
    \label{BostelmannModularOp}
\end{equation}
Note that this reduces to \eqref{Bostelmannvacuumflow} for $\beta\rightarrow \infty$.

\subsubsection*{The interval $(-\infty,t)$}
From this last expression, and conjugating with the (vacuum) translation $U(t)$ as in \eqref{Imodular}, one can compute the modular Hamiltonian associated to $(-\infty,t)$, 
\begin{equation}
    K_{H((-\infty,t))}f(x)=\beta i^\oplus\left(1-e^{\frac{2\pi}{\beta} (x-t)}\right)f'(x)  ,
    \label{BostelmannModHamiltonian}
\end{equation}
and the relative entropy
\begin{equation}
    S_{R((-\infty,t))}(\omega_f||\omega)=\beta \int_{-\infty}^t \left(1-e^{\frac{2\pi}{\beta}(x-t)}\right)f'(x)^2\,dx.
    \label{BostelmannBetaRelEntropy}
\end{equation}
which is was first computed in \cite{Bostelmann:2020srs}. Note that in Proposition 5.6 of that reference it is shown that the subspace $H((-\infty,t))$ is both standard and factorial, so the relative entropy can be computed with \eqref{relativeentropy}.

\subsubsection*{The interval $(a,b)$}

Now we would like to approach the computation of the relative entropy for the bounded interval $I=(a,b)$. The strategy is analogous to the vacuum case of the previous subsection, but three issues are worth mentioning. First, the subspaces $H(I)$ must be shown to be standard and factorial (which we do at the end). Second the assignment $I \mapsto H(I)$ is not PSL$(2,\mathbb{R})$-covariant anymore, namely $U_\beta(g)H(I)\neq H(g\cdot I)$. However, we only need to find a $\bar{g}$ such that
\begin{equation}
    U_\beta(\bar{g})H(I)=H(\mathbb{R}_-),
\end{equation}
then we conjugate with this unitary the modular operator of the negative real line \eqref{BostelmannModularOp} (in complete analogy with \eqref{Imodular}). Explicitly,
\begin{equation}
            \Delta_I^{iu}=U(\Bar{g})^{-1}\Delta_{(-\infty,0)}^{iu}U(\Bar{g}),\quad u\in \mathbb{R}.
            \label{Imodular2}
\end{equation} 

Third, the attempt to construct $\bar{g}$ as described in the vacuum case, see Figure \ref{figure:gbar}, is not immediate to generalize, since the unitary rotation is no longer available (the vacuum rotation does not leave $w_2^{(\beta)}$ invariant). In \cite{Borchers:1998ye} the authors find the unitary dilations and translations \eqref{BYoperators}. We need to find a unitary operator that works as a rotation, meaning that \textit{it does not fix} $\infty$ (in the real-line picture). We propose that there exists $\alpha(\theta,x)$  such that 
\begin{equation}
       U_\beta(r(\theta))f(x)=f(\alpha(\theta,x)).
       \label{alphadef}
\end{equation}
This means $\alpha(\theta,x)$ should obey the following three conditions:
\begin{enumerate}
    \item Identity: $\alpha(0,x)=x$ \\
    \item 1-parameter group: $\alpha(\theta_1,\alpha(\theta_2,x))=\alpha(\theta_1+\theta_2,x)$ \\
    \item $w_2^{(\beta)}$-compatibility: $\frac{\partial \alpha(\theta,x)}{\partial x}\frac{\partial \alpha(\theta,y)}{\partial y}\sinh\left(\frac{\pi}{\beta}(\alpha(\theta,x)-\alpha(\theta,y))\right)^{-2}=\sinh\left(\frac{\pi}{\beta}(x-y)\right)^{-2}$ 
\end{enumerate}
Of course, $\alpha(\theta,x)$ depends also on $\beta$. The third condition, together with \eqref{Uplus}, assures that the operator $U^\oplus_\beta$ induced by $U_\beta(r(\theta))$ is unitary. Let us see why, 
\begin{align}
    \langle U^\oplus_\beta (f_1\oplus 0 + i^\oplus g_1 \oplus 0)\,&,\,U^\oplus_\beta (f_2\oplus 0 + i^\oplus g_2 \oplus 0) \rangle^\oplus\nonumber\\
    &=\langle U_\beta(r(\theta)) f_1\oplus 0 + i^\oplus U_\beta(r(\theta)) g_1 \oplus 0 \,,\,U_\beta(r(\theta)) f_2\oplus 0 + i^\oplus U_\beta(r(\theta)) g_2 \oplus 0 \rangle^\oplus \nonumber \\
    &= w_2^{(\beta)}(f_1\circ \alpha(\theta,\cdot),f_2\circ \alpha(\theta,\cdot))+w_2^{(\beta)}(g_1\circ \alpha(\theta,\cdot),g_2\circ \alpha(\theta,\cdot))\nonumber\\
    &+ \tau_\beta(f_1\circ \alpha(\theta,\cdot),-D g_2\circ \alpha(\theta,\cdot))+i\tau_\beta(f_1\circ \alpha(\theta,\cdot), g_2\circ \alpha(\theta,\cdot))\nonumber\\
    &+ \tau_\beta(-D g_1\circ \alpha(\theta,\cdot), f_2\circ \alpha(\theta,\cdot))-i\tau_\beta(g_1\circ \alpha(\theta,\cdot), f_2\circ \alpha(\theta,\cdot))\nonumber\\
    &= w_2^{(\beta)}(f_1\circ \alpha(\theta,\cdot),f_2\circ \alpha(\theta,\cdot))+w_2^{(\beta)}(g_1\circ \alpha(\theta,\cdot),g_2\circ \alpha(\theta,\cdot))\nonumber\\
    &- \sigma(f_1\circ \alpha(\theta,\cdot), g_2\circ \alpha(\theta,\cdot))+i\tau_\beta(f_1\circ \alpha(\theta,\cdot), g_2\circ \alpha(\theta,\cdot))\nonumber\\
    &+ \sigma(g_1\circ \alpha(\theta,\cdot), f_2\circ \alpha(\theta,\cdot))-i\tau_\beta(g_1\circ \alpha(\theta,\cdot), f_2\circ \alpha(\theta,\cdot))\nonumber\\
    &= w_2^{(\beta)}(f_1\circ \alpha(\theta,\cdot),f_2\circ \alpha(\theta,\cdot))+w_2^{(\beta)}(g_1\circ \alpha(\theta,\cdot),g_2\circ \alpha(\theta,\cdot))\nonumber\\
    &+i w_2^{(\beta)}(f_1\circ \alpha(\theta,\cdot), g_2\circ \alpha(\theta,\cdot))-iw_2^{(\beta)}(g_1\circ \alpha(\theta,\cdot), f_2\circ \alpha(\theta,\cdot))\nonumber\\
    &=w_2^{(\beta)}(f_1,f_2)+w_2^{(\beta)}(g_1,g_2)+i w_2^{(\beta)}(f_1, g_2)-iw_2^{(\beta)}(g_1, f_2)\nonumber\\
    &=\langle f_1\oplus 0 + i^\oplus g_1 \oplus 0\,,\,f_2\oplus 0 + i^\oplus g_2 \oplus 0\rangle^\oplus,
\end{align}
where in the fifth equality we used the $w_2^{(\beta)}$-compatibility condition.

In order to find $\alpha(\theta,x)$, there is a hint coming
from the PSL$(2,\mathbb{R})$ product rules: 
\begin{align*}
      {\delta(s)r(\theta)\delta(-s)=r\left(2\arctan\left( e^{-s}\lambda\right)\right)\delta\left(\log\left[\frac{1+e^{-2s}\lambda^2}{1+\lambda^2}\right]\right) \tau\left(\frac{2\sinh(s)\lambda}{1+e^{-2s}\lambda^2}\right)}
\end{align*}
where $\lambda=\tan\frac{\theta}{2}$.
This translates, by means of \eqref{BYoperators} and \eqref{alphadef},  into a functional equation:
\begin{equation}
    \left(e^s+e^{-s}\lambda^2\right)\left(e^{\frac{2\pi}{\beta}\alpha\left(\theta,\phi(s,x)\right) }-1  \right)=\left(1+\lambda^2\right)\left(e^{\frac{2\pi}{\beta} \alpha(2\arctan\left( e^{-s}\lambda\right),x)}-1\right)-\frac{4\pi}{\beta}\sinh(s)\lambda 
\end{equation}
with
\begin{equation}
    \phi(s,x)=\frac{\beta}{2\pi}\log\left(1+e^{-s}(e^{\frac{2\pi x}{\beta}}-1)\right) 
\end{equation}
It is convenient to work with $A(\lambda,x)$ defined by:
      \begin{equation}
          \alpha(\theta,x)=\frac{\beta}{2\pi} \log\left[1+A(\lambda,x)\right]
          \label{Adef}
      \end{equation}
Differentiating w.r.t. $s$ and setting $s=0$ we get a PDE
\begin{equation}
    \frac{\lambda^2-1}{\lambda^2+1}A(\lambda,x)+\frac{\beta}{2\pi}\left(1-e^{-\frac{2\pi}{\beta}x}\right)\partial_xA(\lambda,x)=\lambda \partial_\lambda A(\lambda,x)+\frac{4\pi}{\beta}\frac{\lambda}{\lambda^2+1},
\end{equation}
which has infinite solutions of the form
\begin{equation}
    A(\lambda,x)=\frac{2\pi}{\beta} \left[-\lambda +\frac{1+\lambda^2}{\lambda} B\left(\lambda \left(e^{\frac{2\pi}{\beta }x}-1\right)\right)\right],
    \label{Asol}
\end{equation}
for any function $B$. From $\alpha(0,x)=x$ we get 
      \begin{equation}
          B(z)\sim\frac{\beta}{2\pi}z,\qquad  z\rightarrow 0.
      \end{equation}
From this and condition 2 above (group property) evaluated at $x=0$ we get 
\begin{equation}
     B(z)=\frac{z}{z+\frac{2\pi}{\beta}}
\end{equation}
Now plugging this form of $B$ into \eqref{Asol} and taking into account \eqref{alphadef} and \eqref{Adef},  
$U_\beta\left(r(\theta)\right)f=f(\alpha(\theta,\cdot))$ can be shown to be compatible with $w_2^{(\beta)}$ (condition 3 above) where
\begin{equation}
    \alpha(\theta,x)=\frac{\beta}{2\pi} \log\left[1+A(\lambda,x)\right],\qquad A(\lambda,x)=\frac{2\pi}{\beta} \frac{e^{\frac{2\pi}{\beta}x }-1-\frac{2\pi}{\beta}\lambda}{\lambda (e^{\frac{2\pi}{\beta}x}-1)+\frac{2\pi}{\beta}}.
    \label{alphasol}
\end{equation}
Let us see this, first of all we rewrite  the $w_2^{(\beta)}$-compatibility condition,
\begin{equation*}
    \sinh^{2}\left(\frac{\pi}{\beta}(\alpha(\theta,x)-\alpha(\theta,y))\right)=\frac{\partial \alpha(\theta,x)}{\partial x}\frac{\partial \alpha(\theta,y)}{\partial y}  \sinh^{2}\left(\frac{\pi}{\beta}(x-y)\right).
\end{equation*}
A straightforward computation (using $e^{\frac{\pi}{\beta}\alpha(\theta,x)}=(1+A(\lambda,x))^{\frac{1}{2}}$) of the square root of the left hand side gives
\begin{equation*}
    \sinh\left(\frac{\pi}{\beta}(\alpha(\theta,x)-\alpha(\theta,y))\right)=
    \frac{1}{2}\left[ \frac{(1+A(\lambda,x))^{\frac{1}{2}}}{(1+A(\lambda,y))^{\frac{1}{2}}}-\frac{(1+A(\lambda,y))^{\frac{1}{2}}}{(1+A(\lambda,x))^{\frac{1}{2}}} \right].
\end{equation*} 
Squaring this expression and with \eqref{alphasol},
\begin{align}
    \sinh^{2}&\left(\frac{\pi}{\beta}(\alpha(\theta,x)-\alpha(\theta,y))\right)=
    \frac{1}{4}\left[ \frac{1+A(\lambda,x)}{1+A(\lambda,y)}-\frac{1+A(\lambda,y)}{1+A(\lambda,x)} -2\right] \nonumber\\
    &=Y(\lambda,x)Y(\lambda,y) \sinh^2\left(\frac{\pi}{\beta}(x-y)\right)
\end{align}
where 
\begin{align*}
    Y(\lambda,x)&=\frac{4 \pi ^2 \beta \left(\lambda ^2+1\right) e^{\frac{2 \pi  x}{\beta }}}{\left(\beta  \lambda  \left(e^{\frac{2 \pi  x}{\beta }}-1\right)+2 \pi \right) \left(\beta  (\beta  \lambda +2 \pi ) e^{\frac{2 \pi  x}{\beta }}-\left(\beta ^2+4 \pi ^2\right) \lambda \right)}.
\end{align*}
But it turns out that a straightforward computation gives
\[\frac{\partial \alpha}{\partial x}(\theta,x)=Y(\lambda,x)\]
which means that the $w_2^{(\beta)}$-compatibility condition holds.

Having found a unitary rotation, we can implement the first transformation of Figure \ref{figure:gbar} with  $U_\beta(r(\Tilde{\theta}))$ where
        \begin{equation}
           \Tilde{\theta}=2\arctan\left(-\frac{2\pi}{\beta}\frac{e^{\frac{2\pi a}{\beta}} }{e^{\frac{2\pi a}{\beta}}-(\frac{2\pi}{\beta})^2-1}\right),
            \label{tildetheta}
        \end{equation}
and taking into account that the corresponding unitary $U_\beta^\oplus(r(\theta))$ on $\mathcal{H}^\oplus$ is defined by \eqref{Uplus}. This rotation sends $a$ to $-\infty$ and $b$ to $b'=\alpha(\Tilde{\theta},b)$, so it maps $H(I)$ to $H((-\infty,b'))$. Then, by a unitary vacuum translation $U(-b')$,  $H((-\infty,b'))$ is mapped to $H(\mathbb{R}_-)$ as desired. The unitary $U(\bar{g})$ is the composition of these two unitary transformations. 

From \eqref{Imodular2}, the modular evolution on $H(I)\oplus 0$ is
\begin{equation}
    \Delta^{iu}_{(a,b)}f(x)\oplus 0=f\left[\frac{\beta}{2\pi}\log\left( 
    \frac{\sinh(\pi u)e^{\frac{\pi}{\beta}(a+b-x)}-\sinh(\frac{\pi}{\beta}(b-a)+\pi u)e^{\frac{\pi x}{\beta}}}{-\sinh(\pi u)e^{-\frac{\pi}{\beta}(a+b-x)}+\sinh(-\frac{\pi}{\beta}(b-a)+\pi u)e^{-\frac{\pi x}{\beta}}}   \right)
    \right]\oplus 0
    \label{IBetaModularOp}
\end{equation}
By differentiating, the modular Hamiltonian is
     \begin{equation}
     K_{(a,
     b)}f(x)\oplus 0=2 \beta i^\oplus \dfrac{ \sinh{(\frac{\pi}{\beta}(x-a))}\sinh{(\frac{\pi}{\beta}(b-x))}}{\sinh{(\frac{\pi}{\beta}(b-a))}}f'(x)\oplus 0
     \label{IBetaModHamiltonian}
     \end{equation}
It coincides with \eqref{BostelmannModHamiltonian} in the limit $a\rightarrow -\infty$ and with \eqref{IModHam} for $\beta\rightarrow\infty$. 
Again like in the vacuum case now one can still see that $P_{I}K_{I}f(x)=\chi_{I}(x)K_{I}f(x)$. The proof is similar to the one we showed above, the only difference is that this time we have to see that $g$ has finite $\tau_{\beta}$ norm, but the same argument works. Given $f\in C^{\infty}_{c}(\mathbb{R})$ define once again the functions $g(x)=\chi_{I}(x)K_{I}f(x)$ and $g_{c}(x)=\chi_{I^{c}}(x)K_{I}f(x)$, these are piecewise-differentiable functions so its Fourier transform is bounded and decays at least like $p^{-2}$ for large $p$, then
\begin{equation*}
    \|g\|_{\tau_{\beta}}^{2}=\frac{1}{\pi}\Re{\int_{0}^{\infty}\dfrac{|\hat{g}(p)|^2 p}{1-e^{-\beta p}}dp}<\infty.
\end{equation*}     

Finally, the relative entropy is given by
\begin{equation}
    S_{R(I)}(\omega_f||\omega)=2 \beta \int_a^b \dfrac{ \sinh{(\frac{\pi}{\beta}(x-a))}\sinh{(\frac{\pi}{\beta}(b-x))}}{\sinh{(\frac{\pi}{\beta}(b-a))}}f'(x)^2\,dx
    \label{Ibetarelativeentropy}
\end{equation}
which is our main result. This relative entropy coincides, modulo some factor,  with the modular Hamiltonian in the cut-off theory (equation (4.2) in \cite{hartman2015speed}). This can be formally understood by first noticing that the relative entropy can be related to a difference of mean values of the modular Hamiltonian $K$ and a difference of entanglement entropies, 
\begin{equation}
    S(\omega_2||\omega_1)=(\langle K_1\rangle_2-\langle K_1\rangle_1 )- (S_2-S_1).
\end{equation}
In our case the last parenthesis is zero since one state is a unitary applied to the other state (in the vector representation). This explains the connection of \eqref{Ibetarelativeentropy} to the modular Hamiltonian  of \cite{hartman2015speed}. Since the arguments of \cite{hartman2015speed} are of general validity within CFTs, and taking into account the above discussion, it is reasonable to expect that in general \eqref{Ibetarelativeentropy} will hold with $f'(x)^2$ replaced by the classical energy density $T_{00}(x)$ of the theory. We will confirm this expectation in the next section for the massless scalar QFT in 1+1 dimensions.   

Identically to the vacuum case \eqref{vacuumrelativeentropy}, the relative entropy \eqref{Ibetarelativeentropy} is translation invariant and with positive derivative wih respect to the length $L$ of the interval. There is also a Bekenstein-like bound 
\begin{equation}
    S_{(a,
     b)}(f)\leq \pi \frac{L}{2}\,\left( \frac{\tanh\left(\frac{\pi}{\beta}\frac{L}{2}\right) } {\frac{\pi}{\beta}\frac{L}{2}}\right) \int_a^b f'(x)^2\,dx \leq \pi \frac{L}{2}\, \int_a^b f'(x)^2\,dx,
     \label{IbetaQNECbound}
\end{equation}
and a QNEC-like bound \footnote{Again, this  the same as the result of Proposition 3.7 of \cite{Bostelmann:2020srs}, this time  $T_{f}(s,L)$ is $2\beta \int_{-\min\left\{s,\frac{L}{2}\right\}}^{\min\left\{s,\frac{L}{2}\right\}}\dfrac{\sinh{\frac{\pi}{\beta}(x+\frac{L}{2})}\sinh{\frac{\pi}{\beta}(\frac{L}{2}-x)}}{\sinh{\frac{\pi L}{\beta}}}f'(x)^2\,dx$. It is straightforward to see that such $T_{f}$ satisfies the smoothness hypothesis of the Proposition. }
\begin{equation}
    \frac{d^2}{dL^2}S_{I(L)}(f)\geq -\frac{\pi^2}{\beta \sinh^3(L\frac{\pi}{\beta})} \int_a^b \resizebox{.5\hsize}{!}{$\left[\left(\cosh( \frac{\pi}{\beta}L)-1\right)^2+2\sinh^2\left(\frac{\pi}{\beta}(x-c)\right)\left(1+\cosh^2(\frac{\pi}{\beta}L)\right)\right]$}f'(x)^2\,dx,
\end{equation}
where $c=(a+b)/2$. We shall discuss this expression later on.

Before concluding this section we have to show that $H(I)$ is standard and factorial so the machinery we have been using, and in particular \eqref{relativeentropy}, is valid. We do this in the following Proposition.

\begin{proposition}
$H(I)$ is standard and factorial
\begin{proof}
    The condition of separability $H(I)\cap i^\oplus H(I)=0$ follows exactly as in Proposition 5.6 of \cite{Bostelmann:2020srs} (or with the logic for what follows). The cyclicity, $(H(I)+i^\oplus H(I))^\perp=0$, can be shown to hold using the unitary rotation \eqref{alphadef}. Given any $H((a,b))$ there is an associated subspace $H_{b'}:=H((-\infty,b'))=U^\oplus_\beta(r(\Tilde{\theta}))H(I)$ obtained by a rotation in $\tilde{\theta}$ given by \eqref{tildetheta} and explained after that equation. The subspace $H_{b'}$ is, by Proposition 5.6 of \cite{Bostelmann:2020srs}, standard and factorial. It is immediate to show that $0=(H_{b'}+i^\oplus H_{b'})^\perp=U^\oplus_\beta(r(\Tilde{\theta}))(H(I)+i^\oplus H(I))^\perp$ which implies that $H(I)$ is cyclic. Similarly, we conclude that $H(I)$ is factorial since $0=H_{b'}\cap H_{b'}'=U^\oplus_\beta(r(\Tilde{\theta}))(H(I)\cap H(I)')$.
\end{proof}
    
\end{proposition}

\section{The free massless boson in $1+1$ dimensions at finite temperature}

In this section we take advantage of the quantities we have computed for the chiral boson and combine the two chiralities in order to obtain the modular flow, modular Hamiltonian and relative entropy on the interval $(a,b)$ for the massless free boson in two dimensions $\Phi$. 

First of all, let us define $x^\pm=t\pm x$, and $j^\pm(x^\pm)=\partial_\pm\phi^\pm(x^\pm)$ are  the (non-smeared) chiral currents of the previous section (below we give further details). In this section we will use $\pm$ symbols to denote copies of the objects of the chiral case (with the exception of the symplectic structure $\sigma$ and the bilinear form $\tau$). So for example $\mathcal{H}$ now refers to a Hilbert space of the two-dimensional model, and $\mathcal{H}_\pm$ are Hilbert spaces of the chiral case. 

The symplectic space of the massless boson in two dimensions is\footnote{Here we are defining $f\in \dot{C}_c^\infty(\mathbb{R})$ if $f\in {C}_c^\infty(\mathbb{R})$ and $\hat{f}(0)=0$. This is necessary to avoid the well-known IR problem of the massless 2-dimensional field \cite{Streater:1989vi}.} \cite{Longo:2021rag} 
\begin{equation}
    \mathcal{K}=C_c^\infty(\mathbb{R})\oplus \dot{C}_c^\infty(\mathbb{R})
\end{equation} 
with symplectic structure 
\begin{equation}
    \sigma_{2D}((f_1,g_1),(f_2,g_2))=\frac{1}{2}\int_\mathbb{R} dx (g_1(x) f_2(x)-f_1(x) g_2(x))
    \label{symplectic2dim}
\end{equation}
Here the pair $(f,g)\in\mathcal{K}$ should be thought as the initial conditions $\Phi(0,x)=f(x)$, $\dot\Phi(0,x)=g(x)$ of a solution $\Phi(t,x)$ of the Klein-Gordon equation. In general, 
\begin{equation}
    \Phi(t,x)=\phi_+(x^+)+\phi_-(x^-),\qquad \phi^\pm \in C_c^\infty(\mathbb{R}).
\end{equation}
Then, the symplectic structure \eqref{symplectic2dim} can be written as
\begin{equation}
    \sigma_{2D}((f_1,g_1),(f_2,g_2))=-\int_\mathbb{R}dx \left(\phi_1^+(x) \phi_2^+{}'(x)+\phi_1^-(x) \phi_2^-{}'(x) \right)
    \label{symplecticstructures}
\end{equation}
 The lack of mixing between the chiralities implies that there is a symplectic isomorphism\footnote{It is most easily written in Fourier space: $\hat\phi_\pm(\pm p)=\frac{1}{2}(\hat{f}( p)\pm \frac{i}{p} \hat{g}( p))$.} $\chi$ that maps $(\mathcal{K},\sigma_{2D})$ to $(\mathcal{K}_- \oplus \mathcal{K}_+,-(\sigma\oplus\sigma)) $, with inverse given by
\begin{equation}
\chi^{-1}\begin{pmatrix}\phi_+ \\ \phi_-\end{pmatrix}=\begin{pmatrix}\phi_+(x)+\phi_-(-x) \\ {\phi'_+}(x)+{\phi'_-}(-x)\end{pmatrix}=\begin{pmatrix} f(x) \\ g(x) \end{pmatrix}.
    \label{symplectomorphism}
\end{equation}
In turn, this implies that the CCR$(\mathcal{K},\sigma_{2D})$ algebra is equivalent to the tensor product \[\text{CCR}(\mathcal{K}_-,-\sigma)\otimes
\text{CCR}(\mathcal{K}_+,-\sigma),\] 
with $\sigma$ as in \eqref{symplectic}. More precisely, we identify these CCR-algebras by
\begin{equation}
    W(\phi_-(x))\otimes W(\phi_+(x))\mapsto W((\phi_+(x)+\phi_-(-x),\,\,{\phi_+}'(x)+{\phi_-}'(-x)))
\end{equation}
This is in fact a $*-$isomorphism of the algebras. The change in sign in the symplectic structure $\sigma$ w.r.t to the previous section requires a change in sign in the complex structure\footnote{In order to see this, note that $\tau$ is independent of this change in sign, since it must be positive. Therefore from the defining equation of the complex structure $\tau(\cdot,D\cdot)=\sigma(\cdot,\cdot)$ it is seen that a change in sign in $\sigma$ translates into a change in sign in $D$ and therefore in the complex structure.}, and these two signs end up compensating each other in the relative entropy\footnote{The 1-particle modular Hamiltonian $K$ is not affected by this sign change, since $S$ is not affected as seen by its definition and neither is $S^*$, therefore $\Delta=S^*S$ is not affected.} \eqref{relativeentropy}.

Given a positive symmetric bilinear form $\tau_{2D}$ on $\mathcal{K}$ and its corresponding quasi-free state on CCR$(\mathcal{K},\sigma_{2D})$, by the isomorphisms mentioned above we get a quasi-free product state on  CCR$(\mathcal{K}_+,\sigma)\otimes \text{CCR}(\mathcal{K}_-,\sigma)$ with the same $\tau$ for each chiral copy. Therefore the vacuum one-particle Hilbert space is
\begin{equation}
\mathcal{H}\simeq \mathcal{H}_-\oplus\mathcal{H}_+ 
\end{equation}
where $\mathcal{H}_\pm$ are copies of the chiral boson Hilbert space $L^2(\mathbb{R}_+,pdp)$. The isomorphism \eqref{symplectomorphism} is anti-linear, since in momentum space (or coordinate space, using properties of the Hilbert transform $\mathfrak{H}$) it is direct to show that 
\begin{equation}
    \chi^{-1}i_1=-i_2\chi^{-1},
\end{equation}
where $i_1$ is the complex structure of the chiral boson and $i_2$ is the complex structure in \cite{Longo:2020amm,Longo:2021rag} 
\begin{equation}
    i_2:=\left(\begin{matrix}
    0& |p|^{-1}\\
    -|p| &0
    \end{matrix}\right)
\end{equation}
Therefore, 
\begin{align}
    \tau_{2D}(\Phi,\Psi)&= \sigma_{2D}(\Phi,i_2\Psi)\nonumber\\
    &=-\sigma(\phi_+,(\chi i_2\Psi)_+)-\sigma(\phi_-,(\chi i_2\Psi)_-)\nonumber\\
    &=-\tau(\phi_+,-i_1(\chi i_2\Psi)_+)-\tau(\phi_-,-i_1(\chi i_2\Psi)_-)\nonumber\\
    &=\tau(\phi_+,i_1(\chi i_2\Psi)_+)+\tau(\phi_-,i_1(\chi i_2\Psi)_-)\\
    &=\tau(\phi_+,\psi_+)+\tau(\phi_-,\psi_-)
    \label{taus}
\end{align}

Analogously, for the thermal state we have
\begin{equation}
    \mathcal{H}^\oplus\simeq\mathcal{H}_-^\oplus \oplus\mathcal{H}_+^\oplus
\end{equation}
with $\mathcal{H}_\pm^\oplus$ two copies of the purified Hilbert space that we constructed in the previous section (which was called $\mathcal{H}^\oplus$, we hope there is no confusion). The Fock spaces are related as
\begin{equation}
    \Gamma(\mathcal{H}^\oplus)\simeq \Gamma(\mathcal{H}_-^\oplus)\otimes \Gamma(\mathcal{H}_+^\oplus)
\end{equation}
From now on we identify all these spaces with the appropriate isomorphisms.

\subsection{Modular flow and modular Hamiltonian}

Let us consider a causal diamond with base $(a,b)$ on the time-zero surface. Its corresponding standard subspace is $H(\Diamond)$ of pairs $(f,g)\in \mathcal{K}$ supported on the interval $(a,b)$ or equivalently Klein-Gordon fields $\Phi$ with initial conditions given by $(f,g)$.  Note that such  diamond is described  in null coordinates as $(x^-,x^+)\in ((-b,-a), (a,b))$. Right wedges are obtained in limit $b\rightarrow \infty$ and similarly $a\rightarrow -\infty$ for left wedges.

At the one-particle level, we have
\begin{equation}
    K_{H(\Diamond)}\simeq K_{H((-b,-a))} \oplus K_{H((a,b))}.
    \label{Ks}
\end{equation} 
This follows from the fact that for $\Phi, \Psi \in H (\Diamond)$ 
\begin{align}
    S_{H(\Diamond)} (\Phi + i \Psi) &=\Phi - i \Psi \nonumber \\
    &= \phi_+ -i  \psi^+ + \phi_- -i \psi^-  \nonumber\\
    &=S_{H((-b,-a))}(\phi_- +i \psi^-  ) + S_{H((a,b))} (\phi_+ +i  \psi^+)
\end{align}
implying that $S_{H(\Diamond)}\simeq S_{H((-b,-a))} \oplus S_{H((a,b))}$ and then $\Delta_{H(\Diamond)}\simeq \Delta_{H((-b,-a))} \oplus \Delta_{H((a,b))}$. The modular evolution in the diamond, 
\begin{equation}
\Delta_{H(\Diamond)}^{iu}\simeq \Delta_{H((-b,-a))}^{iu} \oplus \Delta_{H((a,b))}^{iu}
\end{equation}
which explicilty reads,
\begin{equation}
   \left[\Delta_{H(\Diamond)}^{iu} \left( \begin{matrix}
        f\\
        g
\end{matrix}\right)\oplus \left( \begin{matrix}
        0\\
        0
\end{matrix}\right)\right](x)=\left(\begin{matrix}
   [\Delta_{H((-b,-a))}^{iu} \phi_-](-x)+[\Delta_{H((a,b))}^{iu} \phi_+](x)\\  [\Delta_{H((-b,-a))}^{iu} \phi_-]'(-x)+[\Delta_{H((a,b))}^{iu} \phi_+]'(x)
    \end{matrix}\right)\oplus \left( \begin{matrix}
        0\\
        0
\end{matrix}\right)
\end{equation}
where $\phi_\pm$ should be thought as given in terms of $(f,g)$ using the isomorphism $\chi$ and the evolution of each chirality is given in \eqref{IBetaModularOp}. A more intuitive presentation of the modular flow is to show the geometric transformation of the coordinates $(t,x)$ inside the diamond, as in 
Figure  \ref{modularevolution}.

\begin{figure}[h]
\centering
\includegraphics[width=11cm]{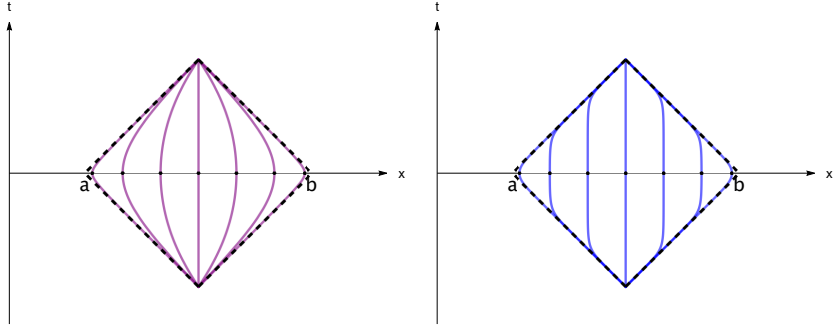}
\caption{Modular flow for low temperature (left) and high temperature (right)}
\label{modularevolution}
\end{figure}

\subsection{Relative entropies and bounds}

The relative entropy in two dimensions is the sum of the relative entropies of the chiral copies, which follows from \eqref{symplecticstructures} and \eqref{Ks}. Before arriving to an explicit expression of the relative entropies for different cases, we first  find the modular Hamiltonians.   

\subsubsection*{The wedge}

On a right wedge $W_R(a)$ with base $(a,\infty)$, given \eqref{symplectomorphism} and \eqref{Ks}, we have the corresponding vacuum modular Hamiltonian acting on the initial conditions
\begin{align}
    K_{H((W_R(a))}\left(\begin{matrix}
f\\
g
    \end{matrix}\right)
    &=\left( \begin{matrix}
    (K_{H((-\infty,-a))}\phi_-)(-x) +(K_{H((a,\infty))}\phi_+)(x)\\
    (K_{H((-\infty,-a))}\phi_-)'(-x) +(K_{H((a,\infty))}\phi_+)'(x)
    \end{matrix}
 \right)\nonumber\\
 &=-2\pi i\left( \begin{matrix}
    (-a+x)\phi_-'(-x) +(x-a)\phi_+'(x)\\
    ((-a-x)\phi'_-)'(-x) +((x-a)\phi'_+)'(x)
    \end{matrix}
 \right)\nonumber\\
 &=-2\pi i\left( \begin{matrix}
    (-a+x)\phi_-'(-x) +(x-a)\phi_+'(x)\\
    ((a-x)\phi'_-(-x) +(x-a)\phi'_+(x))'
    \end{matrix}
 \right)\nonumber\\
 &=-2\pi i\left( \begin{matrix}
    (x-a) g(x)\\
    ((x-a) f'(x))'
    \end{matrix}
 \right)
\end{align}
Where in the second line we made use of the antilinearity between the chiral spaces $\mathcal{H}_\pm$ and $\mathcal{H}$. Then, 
\begin{equation}
    K_{H(W_R(a))}=-2\pi i\left( \begin{matrix}
     0& x-a\\
    \frac{d}{dx}(x-a)\frac{d}{dx} & 0
    \end{matrix}
 \right)
\end{equation}

Plugging this modular Hamiltonian in \eqref{relativeentropy},
\begin{equation}
    S_{H(W_R(a))}((f,g))=2\pi\int_a^\infty  (x-a)T_{00}(x)dx,
\end{equation}
with
\begin{equation}
    T_{00}(x)=\frac{1}{2}(f'(x)^2+g(x)^2)
\end{equation}
the classical energy density at $t=0$ of the KG field $\Phi$. This is the same result as that of \cite{Longo:2019mhx}, with a  translation by $a$.

Similarly, for the thermal state we have
\begin{align}
    K_{H(W_R(a))}^{(\beta)}\left(\begin{matrix}
f\\
g
    \end{matrix}\right)\oplus \left(\begin{matrix}
0\\
0
    \end{matrix}\right) 
     &=-\beta i^\oplus\left( \begin{matrix}
    \left(1-e^{-\frac{2\pi}{\beta}(x-a)}\right) g(x)\\
    \left[\left(1-e^{-\frac{2\pi}{\beta}(x-a)}\right)  f'(x)\right]'
    \end{matrix}
 \right)\oplus \left(\begin{matrix}
0\\
0
    \end{matrix}\right) 
    \label{wedgebetamodularhamiltonian}
\end{align}
Where we have used \eqref{BostelmannModHamiltonian} and  \eqref{IBetaModHamiltonian} in the limit $b\rightarrow\infty$. The relative entropy on the wedge at finite temperature is then,
\begin{equation}
        S_{H(W_R(a))}^{(\beta)}((f,g))=\beta\int_a^\infty  \left(1-e^{-\frac{2\pi}{\beta}(x-a)}\right) T_{00}(x)dx
        \label{wedgebetarelativeentropy}
\end{equation}
Note that this expression is valid even for initial conditions supported outside $x>a$, since the cutting projector in \eqref{relativeentropy} restricts the integral to the wedge \cite{Longo:2020amm}. The only restriction on the initial conditions $(f,g)$ is that they belong to the domain of the modular Hamiltonian \eqref{wedgebetamodularhamiltonian}. 

\subsubsection*{The interval}

Repeating the previous computations for the time-zero interval $(a,b)$, we obtain for the vacuum,
\begin{equation}
    K_{H((a,b))}\left(\begin{matrix}
f\\
g
    \end{matrix}\right) 
     =-2\pi i\left( \begin{matrix}
    \frac{(b-x)(x-a)}{b-a} g(x)\\
    \left[\frac{(b-x)(x-a)}{b-a}  f'(x)\right]'
    \end{matrix}
 \right) .
    \label{Imodularhamiltonian2D}
\end{equation}
The vacuum relative entropy of a coherent state is
\begin{equation}
        S_{H((a,b))}((f,g))=2\pi\int_a^b  \frac{(b-x)(x-a)}{b-a} T_{00}(x)dx
        \label{Irelativeentropy2D}
\end{equation}

On the other hand, at finite temperature we have,
\begin{align}
    K_{H((a,b))}^{\beta}\left(\begin{matrix}
f\\
g
    \end{matrix}\right)\oplus \left(\begin{matrix}
0\\
0
    \end{matrix}\right) 
     &=-2\beta i^\oplus\left( \begin{matrix}
    \dfrac{ \sinh{(\frac{\pi}{\beta}(x-a))}\sinh{(\frac{\pi}{\beta}(b-x))}}{\sinh{(\frac{\pi}{\beta}(b-a))}} g(x)\\
    \left[ \dfrac{ \sinh{(\frac{\pi}{\beta}(x-a))}\sinh{(\frac{\pi}{\beta}(b-x))}}{\sinh{(\frac{\pi}{\beta}(b-a))}}  f'(x)\right]'
    \end{matrix}
 \right)\oplus \left(\begin{matrix}
0\\
0
    \end{matrix}\right) 
    \label{Ibetamodularhamiltonian2D}
\end{align}
The relative entropy of a coherente state in the thermal state representation is,
\begin{equation}
        S_{H((a,b))}^{(\beta)}((f,g))=2\beta\int_a^b  \dfrac{ \sinh{(\frac{\pi}{\beta}(x-a))}\sinh{(\frac{\pi}{\beta}(b-x))}}{\sinh{(\frac{\pi}{\beta}(b-a))}}  T_{00}(x)dx
        \label{Ibetarelativeentropy2D}
\end{equation}
This expression confirms, at least for this model, the expectation that in a CFT the relative entropy of coherent states on a finite interval has this form, where the dependence on the model enters only in $T_{00}$. Because of this, the bounds obtained earlier for the chiral model also hold in this case. The Bekenstein-like bound reads, 
\begin{equation}
    S_{(a,
     b)}((f,g))\leq \pi \frac{L}{2}\,\left( \frac{\tanh\left(\frac{\pi}{\beta}\frac{L}{2}\right) } {\frac{\pi}{\beta}\frac{L}{2}}\right) \int_a^b T_{00}(x)\,dx \leq \pi \frac{L}{2}\, \int_a^b T_{00}(x)\,dx.
\end{equation}
While the QNEC-like bound is,
\begin{equation}
    \frac{d^2}{dL^2}S_{I(L)}((f,g))\geq -\frac{\pi^2}{\beta \sinh^3(L\frac{\pi}{\beta})} \int_a^b \resizebox{.5\hsize}{!}{$\left[\left(\cosh( \frac{\pi}{\beta}L)-1\right)^2+2\sinh^2\left(\frac{\pi}{\beta}(x-c)\right)\left(1+\cosh^2(\frac{\pi}{\beta}L)\right)\right]$}T_{00}(x)\,dx.
\end{equation}

\section{Conclusions}

We have extended the relative entropy on $\mathbb{R}_-$ with $T\geq 0$ of {\footnotesize\color{blue} [BCD '22]} to a bounded interval (see \eqref{Ibetarelativeentropy}). In order to achieve this, we found a unitary in the thermal Hilbert space implementing a rotation. Such unitary may turn out to be useful for other related computations.

From the relative entropy \eqref{Ibetarelativeentropy} a Bekenstein-like bound and a QNEC-like bound can be observed. There is however a violation of the QNEC $S''>0$, and all of this is in agreement with \cite{Bostelmann:2020srs}. For the vacuum case, given an energy $E$ we can find a family of functions $f_{n}\in H(I)$ such that $S_{I}''(f_{n})$  given in 
 \eqref{QNECvacuum} goes to zero (just concentrating  the energy density closer and closer around the center of the interval), thus making the QNEC violation as small as desired. On the contrary,  in the thermal case  this is not possible because there is always a bound for the violation of the QNEC given by
\[ S''_{I}(f_n)\rightarrow -\dfrac{\pi^2}{\beta \sinh^3(L\frac{\pi}{\beta})}\left(\cosh( \frac{\pi}{\beta}L)-1\right)^2 E<0\] despite how the energy density is distributed (see \eqref{IbetaQNECbound}). 
 
The computations in the context of a thermal $U(1)$  current left a clear path to analyse the case of a thermal state of the free massless boson in $1+1$ dimensions  restricted to a causal diamond. In the last Section we obtained the modular Hamiltonian \eqref{Ibetamodularhamiltonian2D} and relative entropy \eqref{Ibetarelativeentropy2D} at finite temperature in 1+1 dimensions, with analogous bounds as in the chiral case. In principle most of these techniques could be used for the massless boson in higher dimensions and also the free massive boson in $d+1$ dimensions with $T>0$ \cite{garbarzpalau}. In addition, it would  be very interesting to extend the formalism to include non-coherent states, although this seems a much more complicated affair. 

\section*{Acknowledgements}

We would like to thank David Blanco and Guillem Pérez-Nadal  and specially both Henning Bostelmann for correspondance regarding  \cite{Bostelmann:2020srs} and Yoh Tanimoto for reading a preliminary version of the manuscript and providing valuable feedback. This work was partially supported by grants PIP and PICT from CONICET and ANPCyT. The work of G.P. is supported by an UBACYT scholarship from the University of Buenos Aires. 

\bibliographystyle{toine}
\bibliography{biblio}{}
\end{document}